\documentclass[10pt,twocolumn,english]{article}
\usepackage{helvet}
\usepackage[latin1]{inputenc}
\usepackage[letterpaper]{geometry}
\usepackage{fancyhdr}
\usepackage{babel}
\usepackage{amsmath}
\usepackage{amssymb}
\usepackage[unicode=true,pdfusetitle,
 bookmarks=true,bookmarksnumbered=false,bookmarksopen=false,
 breaklinks=false,pdfborder={0 0 1},backref=section,colorlinks=false]
 {hyperref}

\makeatletter
\usepackage{geometry}
\usepackage{times,helvet}
\usepackage{graphics,color}
\usepackage{altsty}

\usepackage{babel}
\usepackage{algorithmic}
\usepackage{algorithm}
\newenvironment{proof}{\paragraph{\textit{Proof.}}}{\hfill\null}
\newtheorem{definition}{Definition}[section]
\newtheorem{proposition}{Proposition}[section]
\newtheorem{example}{Example}[section]
\newtheorem{theorem}{Theorem}
\newtheorem{cor}{Corollary}

\begin{document}

\title{Shaping Operations to Attack Robust Terror Networks}

\author{Devon Callahan, Paulo Shakarian\\
Network Science Center and\\
Dept. of Electrical Engineering and Computer Science\\
United States Military Academy\\
West Point, NY 10996\\
Email: devon.callahan[at]usma.edu,\\
 paulo[at]shakarian.net
\and
Jeffrey Nielsen, Anthony N. Johnson\\
Network Science Center and\\
Dept. of Mathematical Science\\
United States Military Academy\\
West Point, NY 10996\\
Email: jeffrey.nielsen[at]usma.edu,\\
 anthony.johnson[at]usma.edu}

\maketitle
\begin{abstract}
\noindent Security organizations often attempt to disrupt terror or insurgent networks by targeting ``high value targets'' (HVT's).  However, there have been numerous examples that illustrate how such networks are able to quickly re-generate leadership after such an operation.  Here, we introduce the notion of a \textit{shaping} operation in which the terrorist network is first targeted for the purpose of reducing its leadership re-generation ability before targeting HVT's.  We look to conduct shaping by maximizing the network-wide degree centrality through node removal.  We formally define this problem and prove solving it is NP-Complete.  We introduce a mixed integer-linear program that solves this problem exactly as well as a greedy heuristic for more practical use.  We implement the greedy heuristic and found in examining five real-world terrorist networks that removing only $12\%$ of nodes can increase the network-wide centrality between $17\%$ and $45\%$.  We also show our algorithm can scale to large social networks of $1,133$ nodes and $5,541$ edges on commodity hardware.
\end{abstract}

\section{Introduction}
Terrorist and insurgent networks are known for their ability to regenerate leadership after targeted attacks.  For example, the infamous Al Qaeda in Iraq terrorist leader Abu Musab al-Zarqawi was killed on June 8th, 2006~\footnote{http://www.nytimes.com/2006/06/08/world/middleeast/08cnd-iraq.html?\_r=1} only to be replaced with Abu Ayyub al-Masri about a week later.~\footnote{http://articles.cnn.com/2006-06-15/world/iraq.main\_1\_al-zarqawi-al-qaeda-leader-zawahiri?\_s=PM:WORLD}  Here, we introduce the notion of a \textit{shaping} operation in which the terrorist network is first targeted for the purpose of reducing its leadership re-generation ability.  Such shaping operations would then be followed by normal attacks against high value targets -- however the network would be less likely to recover due to the initial shaping operations.  In this paper, we look to shape such networks by increasing network-wide centrality, first introduced in \cite{freeman1979centrality}.  Intuitively, this measure provides insight into the criticality of high-degree nodes.  Hence, a network with a low network-wide centrality is a more decentralized organization and likely to regenerate leadership.  In the shaping operations introduced in this paper, we seek to target nodes that will maximize this measure - making follow-on attacks against leadership more effective.  Previous work has primarily dealt with the problem of leadership regeneration by focusing on individuals likely to emerge as new leaders~\cite{carley04}.  However, targeting or obtaining information about certain individuals may not always be possible.  Hence, in this paper, we target nodes that affect the reduce the \textit{network's} ability regenerate leadership as a whole.

The main contributions of this paper is the introduction of a formal problem we call \textit{FRAGILITY} (Section~\ref{prelim-sec}) which seeks to find a set of nodes whose removal would maximize the network-wide centrality.  We also included in the problem a ``no strike list'' - nodes in the network that cannot be targeted for various reasons.  This is because real-world targeting of terrorist or insurgent networks often includes restrictions against certain individuals.  We also prove that this problem is NP-complete (and the associated optimization problem is NP-hard) which means that an efficient algorithm to solve it optimally is currently unknown.  We then provide two algorithms for solving this problem (Section~\ref{algSec}).  Our first algorithm is an integer program that ensures an exact solution and, though intractable by our complexity result, may be amenable to an integer program solver.  Then we introduce a greedy heuristic that we show experimentally (in Section~\ref{expSec}) to provide good results in practice (as we demonstrate on six different real-world terrorist networks) and scales to networks of $1,133$ nodes and $5,541$ edges.  In examining five real-world terrorist networks, we found that successful targetting operations against only $12\%$ (or less) of nodes can increase the network-wide centrality between $17\%$ and $45\%$.  Additionally, we discuss related work further in Section~\ref{rwSec}.

We would like to note that the targeting of individuals in a terrorist or insurgent network does not necessarily mean to that they should be killed.  In fact, for ``shaping operations'' as the ones described in this paper, the killing of certain individuals in the network may be counter-productive.  This is due to the fact that the capture of individuals who are likely emergent leaders may provide further intelligence on the organization in question.

\section{Technical Preliminaries and Computational Complexity}
\label{prelim-sec}
We assume that an undirected social network is represented by the graph $G=(V,E)$.  Additionally, we assume a ``no strike'' set, $S \subseteq V$.  Intuitively, these are nodes in a terrorist/insurgent network that cannot be targeted.  This set is a key part of our framework, as real-world targeting of terrorist and/or insurgents in a terrorist/insurgent network is often accompanied by real-world constraints.  For example, consider the following:
\begin{itemize}
\item We may know an individual's relationships in the terrorist/insurgent network, but may not have enough information (i.e. where he or she may reside, enough evidence, etc.) to actually target him or her.
\item The potential target may be politically sensitive.
\item The potential target may have fled the country or area of operations but still maintains his or her role in the terrorist/insurgent network through electronic communication.
\item The potential ``target'' may actually be a source of intelligence and/or part of an ongoing counter-intelligence operation (i.e. as described in \cite{slowBurn}).
\end{itemize}

Throughout this paper we will also use the following notation.  The symbols $N_G,M_G$ will denote the sizes of $V,E$ respectively.  For each $i \in V$, we will use $d_i$ to denote the degree of that node (the number of individuals he/she is connected to) and $\eta_i$ to denote the set of neighbors and we extend this notation for subsets of $V$ (for $V' \subseteq V, \eta(V')=\bigcup_{i\in V'}\eta_i$).  We will use the notation $\kappa_i$ to denote all edges in $E$ that are adjacent to node $i$ and the notation $d^{*}_G$ to denote the maximum degree of the network.  Given some subset $V'\subseteq V$, we will use the notation $G(V')$ to denote the subgraph of $G$ induced by $V'$.  We describe an example network in Example~\ref{zeroEx}.

\begin{example}
\label{zeroEx}
Consider network $G_{sam}$ in Figure~\ref{ex1fig}.  Nodes $\textbf{a}$ and $\textbf{b}$ may be leaders of a strategic cell that provides guidance to attack cells (nodes $\textbf{c-f}$ and $\textbf{g-j}$).  Note that no members in the attack cells are linked to each other.  Also note that if node $\textbf{a}$ is the leader, and targeted, he could easily be replaced by $\textbf{b}$.
\begin{figure}
    \begin{center}
        \includegraphics[width=.8\linewidth]{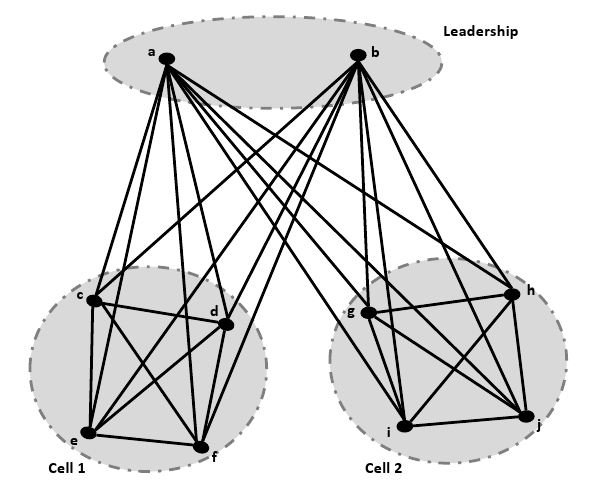}
    \end{center}
\caption{Sample network ($G_{sam}$) for Example~\ref{zeroEx}.}
    \label{ex1fig}
\end{figure}
\end{example}

\subsection{Network-Wide Degree Centrality}

We now introduce the notion of network-wide degree centrality as per \cite{freeman1979centrality}.  The key intuition of this paper is to use this centrality as a measure of the network's ability to re-generate leadership.

\begin{definition}[Network-Wide Degree Centrality~\cite{freeman1979centrality}]
\label{firstDefine}
The degree centrality of a network $G$, denoted $C_G$ is defined as:
\begin{eqnarray}
\label{firstDef}
C_G &=& \frac{\sum_i d^{*}_G-d_i}{(N_G-1)(N_G-2)}
\end{eqnarray}
\end{definition}

We note that there are other types of network-wide centrality (i.e. network-wide betweenness, closeness, etc.).  We leave the consideration of these alternate definitions of network-wide centrality to future work.  Freeman~\cite{freeman1979centrality} shows that for a star network, the quantity $\sum_i d^{*}_G-d_i$ equals $(N_G-1)(N_G-2)$ - and this is the maximum possible value for this quantity.  Hence, the value for $C_G$ can be at most $1$.  As this equation is clearly always positive, network-wide degree centrality is a scalar in $[0,1]$.  Turning back to Example~\ref{zeroEx}, we can compute $C_{G_{sam}}=0.38$ - which seems to indicate that in this particular terrorist/insurgent network that, after leadership is targeted, there is a cadre of second-tier individuals who can eventually take control of the organization.  Throughout this paper, we find it useful to manipulate Equation~\ref{firstDef} as follows.
\begin{eqnarray}
\label{manipEqn}
C_G &=& \frac{N_G d^{*}_G- 2 M_G}{(N_G-1)(N_G-2)}
\end{eqnarray}

We notice that the centrality of a network really depends on three things: number of nodes, number of edges, and the highest degree of any node in the network.  We leverage this re-arranged equation in many of our proofs.  Further, we will use the function $fragile_G : V \rightarrow \Re$ to denote the level of network-wide of the graph after some set of nodes is removed.  Hence, $fragile_G(V')=C_{G(V-V')}$.  We note that this function has some interesting characteristics.  For example, for some subset $V' \subset V$ and element $i \in V-V'$, it is possible that $fragile_G(V') > fragile(V' \cup \{i\})$ or $fragile_G(V') < fragile(V' \cup \{i\})$, hence $fragile_G$ is not necessarily monotonic or anti-monotonic in this sense.  Further, given some additional element $j \in V-V'$, it is possible that $fragile_G(V' \cup \{j\})-fragile_G(V') > fragile_G(V' \cup \{i,j\})-fragile_G(V'\cup\{j\})$ or $fragile_G(V' \cup \{j\})-fragile_G(V') < fragile_G(V' \cup \{i,j\})-fragile_G(V'\cup\{j\})$.  Hence, $fragile_G$ is not necessarily sub- or super- modular either.  Consider Example~\ref{ex1}.

\begin{example}
\label{ex1}
Consider the network $G_{sam}$ in Figure~\ref{ex1fig}.  Here, $fragile_{G_{sam}}(\emptyset)=0.33$, $fragile_{G_{sam}}(\{a\})=fragile_{G_{sam}}(\{b\})=0.57$,  $fragile_{G_{sam}}(\{c\})=0.30$, and $fragile_{G_{sam}}(\{a,b\})=0.0$.  The fact that  $fragile_{G_{sam}}(\{c\}< fragile_{G_{sam}}(\emptyset)$ and $fragile_{G_{sam}}(\{a\} > fragile_{G_{sam}}(\emptyset)$ illustrate that $fragile_{G_{sam}}$ is not necessarily monotonic or anti-monotonic.  Now let us consider the incremental increase of adding an additional element.  Adding $a$ to $\emptyset$ causes $fragile_{G_{sam}}$ to increase by $0.24$ while adding $a$ to $\{b\}\supset \emptyset$ causes $fragile_{G_{sam}}$ to decrease by $0.57$ - implying sub-modularity.  However, adding $c$ to $\emptyset$ causes $fragile_{G_{sam}}$ to decrease by $0.03$ while adding $c$ to set $\{a,b\}\supset \emptyset$ causes $fragile_{G_{sam}}$ to increase by $0.1$ (as $fragile_{G_{sam}}(\{a,b,c\}=0.1$) - implying super-modularity.  Hence, $fragile_{G_{sam}}$ is not necessarily sub- or super- modular.
\end{example}

\subsection{Problems and Complexity Results}

We now have all the pieces to introduce our problems of interest.  We include decision and optimization versions.\\

\noindent\textbf{$FRAGILITY(k,x,G,S)$:}\\
\noindent INPUT: Natural number $k$, real number $x$, network $G=(V,E)$, and no-strike set $S$\\
\noindent OUTPUT: ``Yes'' if there exists set $V' \subseteq V-S$ s.t. $|V'| \leq k$ and $fragile_G(V') > x$ -- ``no'' otherwise.\\

\noindent\textbf{$FRAGILITY\_OPT(k,G,S)$:}\\
\noindent INPUT: Natural number $k$, network $G=(V,E)$, and no-strike set $S$\\
\noindent OUTPUT: Set $V' \subseteq V-S$ s.t. $|V'| \leq k$ s.t. $\not\exists V'' \subseteq V-S$ s.t. $|V''|\leq k$ and $fragile_G(V'') > fragile_G(V')$.\\

As our problems seek to find sets of nodes, rather than individual ones, it raises the question of ``how difficult are these problems.''  We prove that $FRAGILITY$ is NP-Complete - meaning an efficient algorithm to solve it optimally is currently unknown.  Following directly from this result is the NP-hardness of $FRAGILITY\_OPT$.  Below we state and prove this result.

\begin{theorem}[Complexity of $FRAGILITY$]
\label{fNpc}
$FRAGILITY$ is NP-Complete.
\end{theorem}
\begin{proof}
Membership in NP is trivial, consider a set $V'$ of size $k$, -- clearly we can calculate $fragile_G(V')$ in polynomial time.\\
Next we consider the vertex-cover ($VC$) problem and show that it can be embedded into an instance of $FRAGILITY$.  In the $VC$ problem, the input consists of undirected graph $G^{*}=(V^{*},E^{*})$ and natural number $k$.  The output is ``yes'' iff there is a set $V^{**}\subseteq V^{*}$ of size at most $k$ s.t. for all $(i,j) \in E^{*}$, either $i$ or $j$ (or both) are in $V^{*}$.  This problem is well-known to be NP-hard.  First we create a new network $G=(V,E)$ which consists of graph $G^{*}$ but with $N_{G^{*}}+2$ additional nodes which form a star that is disconnected from the rest of the network.  All of the new nodes are put in the no-strike set $S$ (part of the input of $FRAGILITY$).  Clearly, the center of this star is always the most central node in the graph, no matter what is removed from set $V-S$.  This allows us to treat $d^{*}_G$ as a constant equal to $N_G+1$.  Also note that with this construction, for both problems, if a solution exists of less than size $k$, there also exists a solution of exactly size $k$.  Further, we note that for any subset of $V$ whose removal does not affect the overall maximal degree of the network (which is any node outside the set $S$ - hence in some corresponding subset of $V^{*}$ in the graph of the dominating set problem), when some set $V'$ (of size $k$) is removed from $V$, the network-wide degree centrality for the resulting graph can be expressed as follows: $fragile_G(V')=\frac{(N_G-k)(N_G+1)-2(M_G-|\bigcup_{i\in V'}\kappa_i|)}{(N_G-k-1)(N_G-k-2)}$.\\
The proof of correctness of the embedding rests on proving that a ``yes'' answer is returned for the vertex cover problem iff\\ $FRAGILITY(k,\frac{(N_G-k)(N_G+1)-2N_{G^{*}}-2}{(N_G-k-1)(N_G-k-2)},G,S)=\textit{``yes''}$.\\
First, suppose by way of contradiction (BWOC) there is a ``yes'' answer to the $VC$ problem and a ``no'' answer to the corresponding $FRAGILITY$ problem.  Let $V^{**}$ be the set of nodes that cause a ``yes'' answer to $VC$.  If we remove the corresponding nodes from $G$, there are $N_{G^{*}}+1$ edges left in that network.  Hence, as this is a set of size $k$ (thus, meeting the cardinality requirement of $FRAGILITY$ then $fragile_G(V^{**})=\frac{(N_G-k)(N_G+1)-2N_{G^{*}}-4}{(N_G-k-1)(N_G-k-2)}$ which would cause a ``yes'' answer for $FRAGILITY$ -- hence a contradiction.\\
Going the other direction, suppose BWOC there is a ``yes'' answer to the $FRAGILITY$ problem and a ``no'' answer to the corresponding $VC$ problem.  Let $V'$ be the nodes in the solution to $FRAGILITY$.  Clearly, this set is of size $k$ and by how we set up the no-strike list ($S$), there are corresponding nodes in $G^{**}$.  As these nodes cause a ``yes'' answer to $FRAGILITY$, they result in the removal of $M_{G^{*}}$ number of edges in $G$.  By the construction, none of these edges are adjacent to nodes in $S$.  Hence, there are corresponding edges in $G^{*}$.  As this is also the number of edges in $G^{*}$, then this set is also a vertex cover - hence a contradiction.  Hence, as we have shown membership in NP and that this problem is at least as hard as the dominating set problem (resulting in NP hardness), the statement of the theorem follows.
\end{proof}

\begin{cor}[Hardness of $FRAGILITY\_OPT$]
$FRAGILITY\_OPT$ is NP-hard
\end{cor}
\begin{proof}
Follows directly from Theorem~\ref{fNpc}.
\end{proof}

\section{Algorithms}
\label{algSec}
Now with the problems and their complexity identified, we proceed to develop algorithms to solve them.  First, we develop an integer program that, if solved exactly, will produce an optimal solution.  We note that solving a general integer program is also NP-hard.  Hence, an exact solution will likely take exponential time.  However, good approximation techniques such as branch-and-bound exist and mature tools such as QSopt and CPLEX can readily take and approximate solutions to integer programs.  We follow our integer program formulation with a greedy heuristic.  Though we cannot guarantee that the greedy heuristic provides an optimal solution, it often provides a natural approach to approximating many NP-hard optimization problems.

\subsection{Integer Program}

Our first algorithm is presented in the form of an integer program.  The idea is that certain variables in the integer program correspond with the nodes in the original network that can be set to either $0$ or $1$.  An objective function, which mirrors the $fragile$ function is then maximized.  When this function is maximized, all nodes associated with a $1$ variable are picked as the solution.

\begin{definition}[$FRAGILITY\_IP$]
For each $i \in V$, create variables $X_i,Z_i$.  For each undirected edge $ij \in E$, create three variables: $Y_{ij}, Q_{ij}, Q_{ji}$.  Note that the edge is considered in only ``one direction'' for the $Y$ variables and both directions for the $Q$ variables.  We define the $FRAGILITY\_IP$ integer program as follows:

\begin{eqnarray*}
\max& \frac{(N_G - \sum_i X_i)\sum_{ij}Q_{ij}-2\sum_{ij}Y_{ij}}{(N_G -1- \sum_i X_i)(N_G -2- \sum_i X_i)}
\end{eqnarray*}
Subject to:
\begin{eqnarray}
\label{kConst}
& \sum_i X_i \leq k\\
\label{justOneMax}
& \sum_i Z_i = 1\\
\label{adjEdge1}
\forall ij \in E& Y_{ij} \leq 1-X_i\\
\label{adjEdge2}
\forall ij \in E& Y_{ij} \leq 1-X_j\\
\label{adjEdge3}
\forall ij \in E& Q_{ij} \leq Y_{ij}\\
\label{adjEdge4}
\forall ij \in E& Q_{ij} \leq Y_{ji}\\
\forall ij \in E& Q_{ij} \leq Z_i\\
\label{zeroOneMax}
\forall i \in V& Z_i \in \{0,1\}\\
\label{noSconst}
\forall i \in S & X_i = 0\\
\label{zeroOneConst}
\forall i \in V-S& X_i \in \{0,1\}
\end{eqnarray}
\end{definition}

Next we prove how many variables and constraints $FRAGILITY\_IP$ requires as well as prove that it provides a correct solution to $FRAGILITY\_OPT$.

\begin{proposition}
$FRAGILITY\_IP$ has $2 N_G + 3 M_G$ variables and $2+2 N_G + 5 M_G$ constraints.
\end{proposition}

\begin{proposition}
\textit{(1.)} Given the vector $X$ returned by $FRAGILITY\_IP$, the set $\bigcup_{X_i = 1}i$ is a solution to $FRAGILITY\_OPT$.\\
\textit{(2.)}Given a solution $V'$ to $FRAGILITY\_OPT$, $\forall i \in S, X=1$ and  $\forall i \notin S, X=0$ will maximize $FRAGILITY\_IP$.  
\end{proposition}
\begin{proof}
\textit{(1.)} Suppose, BWOC, $\bigcup_{X_i = 1}i$ is not an optimal solution to $FRAGILITY\_OPT$.  Then there is some $V' \neq \bigcup_{X_i = 1}i$ that is.  Suppose $\forall i \in S, X=1$ and  $\forall i \notin S, X=0$.  Clearly, by the definition of a solution to $FRAGILITY\_OPT$, constraints~\ref{kConst},\ref{noSconst} and \ref{zeroOneConst} are all met.  Constraints \ref{adjEdge1} and \ref{adjEdge2} set variables associated with edges adjacent to nodes not in $V'$ to $1$.  Hence, the quantity $\sum_{ij}Y_{ij}$ is equal to the number of edges in the network.  The $Y$ edge variables (both of them for each edge) are also set in a similar manner.  Constraints~\ref{justOneMax},\ref{zeroOneMax} ensures that only one set of such edge variables are set to $1$.  Hence, the quantity $\sum_i X_i)\sum_{ij}Q_{ij}$ is the degree of one node in the network.  As this quantity is present in the objective function and non-negative, it corresponds to the $d^{*}_G$.  As we note that $\sum_i X_i$ is equal to the number of nodes in $G$ when $V'$ is removed, we see that this function is $fragile_G$.  As this quantity is maximized, we have a contradiction.\\
\textit{(2.)} Suppose, BWOC, $\forall i \in S, X=1$ and  $\forall i \notin S, X=0$ is not an optimal solution to $FRAGILITY\_IP$.  Using the same line of reasoning as above, we see that the objective function of $FRAGILITY\_IP$ is the same as $fragile_G$, which also gives us a contradiction.
\end{proof}

Note that this integer program does not have a linear objective function.  However, this can be accommodated for by instead solving $k$ different integer programs and taking the solution from whichever one returns the greatest value for the objective function (that is greater than the initial network-wide degree centrality, of course).  In this case, each integer program is identified with a natural number $i \in \{1,\ldots, k\}$ and the $i$th integer program has the following objective function:
\begin{eqnarray}
\max& \frac{(N_G - i)\sum_{ij}Q_{ij}-2\sum_{ij}Y_{ij}}{(N_G -1- i)(N_G -2- i)}
\end{eqnarray}
As well as constraint~\ref{kConst} as follows:
\begin{eqnarray}
& \sum_i X_i \leq i
\end{eqnarray}

Notice that now the quantities $(N_G - i)$ and $(N_G -1- i)(N_G -2- i)$ can be treated as constants, making the objective function linear.  However, for networks with a heterogeneous degree distribution where $N_G >> k$, it is likely that only the integer program for the case where $i=k$ is needed as removing any node with edges that is unconnected to a maximal degree node will result in an increase in network-wide degree centrality.

Again, we stress that $FRAGILITY\_IP$ provides an exact solution.  As integer-programming is also NP-hard, solving these constraints is likely intractable unless $P=NP$.  However, techniques such as branch-and-bound and mature solvers such as QSopt and CPLEX can provide good approximate solutions to such constraints.  Even if the integer program must be linear, we can use the techniques described above to solve $k$ smaller integer programs or obtaining an approximation by treating the terms involving the total number of nodes in the resulting graph (in the objective function) as constants.  Additionally, a relaxation of the above constraints where $Z_i$ and $X_i$ variables lie in the interval $[0,1]$ is solvable in polynomial time and would provide a lower-bound on the solution to the problem (although this would likely be a loose bound in many cases).

\subsection{A Greedy Heuristic}

The integer program introduced in the last section can be leveraged by an integer-program solver for an approximate solution to $FRAGILITY\_OPT$.  However, it likely will not scale well to extremely large networks.  Therefore, we introduce a greedy heuristic to find an approximate solution.  The ideas is to iteratively pick the node in the network that provides the greatest increase in $fragile$ - and does not cause a decrease.  

\algsetup{indent=1em}
	\begin{algorithm}[!ht]
		\caption{ \textsf{GREEDY\_FRAGILE}}
		\begin{algorithmic}[1]

		\REQUIRE Network $G=(V,E)$, no-strike set $S \subseteq V$, cardinality constraint $k$
		\ENSURE Subset $ V' $
		\medskip

		\STATE{ $V' = \emptyset$}
		\STATE{ $flag = TRUE$}
		
		\WHILE{ $|V'| \leq k$ and $flag$ }
			\STATE{ $curBest = null$, $curBestScore=0$, $haveValidScore=FALSE$ }
			\FOR{ $i \in V-(V'\cup S)$ }
				\STATE{ $curScore = fragile_G(V'\cup\{i\})-fragile_G(V')$}
				\IF{$curScore \geq curBestScore$}
					\STATE{ $curBest = i$}
					\STATE{ $curBestScore = curScore$}
					\STATE{ $haveValidScore=TRUE$}
				\ENDIF
			\ENDFOR
			\IF{ $haveValidScore = FALSE$}
				\STATE{ $flag=FALSE$}
			\ELSE
				\STATE{ $V' = V' \cup \{curBest \}$}
			\ENDIF
		\ENDWHILE
		\RETURN{ $ V' $}.
	\end{algorithmic}
\end{algorithm}

The following two propositions describe characteristics of the output and run-time of $GREEDY\_FRAGILE$, respectively.

\begin{proposition}
If $GREEDY\_FRAGILE$ returns a non-empty solution ($V'$), then $|V'| \leq k$ and $fragile_G(V') > fragile_G(\emptyset)$.
\end{proposition}
\begin{proof}
As the algorithm terminates its main loop once the cardinality of the solution reaches $k$ and as in each iteration, the variable $curBestScore$ is initialized as zero, the statement follows.
\end{proof}

\begin{proposition}
$GREEDY\_FRAGILE$ runs in $O(k N_G^2)$ time.
\end{proposition}
\begin{proof}
We note that $fragile$ is computed in $O(N_G)$ time as it must update the node with the maximum degree.  As the outer loop of the algorithm iterates at most $k$ times and the inner loop iterates $N_G$ times, the statement follows.
\end{proof}

Though our guarantees on $GREEDY\_FRAGILE$ are limited, we show that it performs well experimentally in the next section.

\begin{example}
\label{greedy-ex}
Following from Examples~\ref{zeroEx}-\ref{ex1} using the terrorist/insurgent network $G_{sam}$ from Figure~\ref{ex1fig}, suppose a user wants to identify $3$ nodes that will cause the network to become ``as fragile as possible'' and is able to target any node.  Hence, he would like to solve $FRAGILE\_OPT(3,G_{sam},\emptyset)$ and decides to do so using $GREEDY\_FRAGILE$.  Initially, $fragile_{G_{sam}}(\emptyset)=0.33$.  In the first iteration, it selects and removes node \textbf{a}, increasing the fragility ($fragile_{G_{sam}}(\{a\})=0.57$).  In the next iteration, it selects node \textbf{j}, giving us $fragile_{G_{sam}}(\{a,j\})=0.57$.  Finally, in the third iteration, it picks node \textbf{c}.  This results in $fragile_{G_{sam}}(\{a,j,c\})=0.6$.  The algorithm then terminates.
\end{example}

\section{Implementation and Experiments}
\label{expSec}
All experiments were run on a computer equipped with an Intel Core 2 Duo CPU T9550 processor operating at $2.66$ GHz (only one core was used).  The machine was running Microsoft Windows 7 (32 bit) and equipped with $4.0$ GB of physical memory.  We implemented the $\textsf{GREEDY}\_\textsf{FRAGILE}$ algorithm using Python 2.6 in under $30$ lines of code that leveraged the NetworkX library available from http://networkx.lanl.gov/.

We compared the results of the $\textsf{GREEDY}\_\textsf{FRAGILE}$ to three other more traditional approaches to targeting that rely on centrality measures from the literature.  Specifically, we look at the top closeness and betweenness nodes in the network.  Given node $i$, its closeness is the inverse of the average shortest path length from node $i$ to all other nodes in the graph.  \textit{Betweenness}, on the other hand, is defined as the number of shortest paths between node pairs that pass through $i$.  Formal definitions of both of these measures can be found in \cite{wasserman1994social}.

\subsection{Datasets}

We studied the effects of our algorithm on five different datasets.  The network \textbf{Tanzania}~\cite{moon:2008uq} is a social network of the individuals involved with the Al Qaeda bombing of the U.S. embassy in Dar es Salaam in 1998.  It  was collected from newspaper accounts by subject matter experts in the field.  The remainder networks, \textbf{GenTerrorNw1}-\textbf{GenTerrorNw4} are terrorist networks generated from real-world classified datasets\cite{carley:2009fk,carley:2012kx}.  The \textbf{Tanzania} and the \textbf{GenTerrorNw1}-\textbf{GenTerrorNw4} datasets used in our analysis were
multi-modal networks, meaning they contain multiple node classes such as Agents, Resources, Locations, etc. The presence of the different node classes generate multiple or meta networks, which, in their original state, do not provide the single-mode Agent by Agent network needed to test our algorithms.  Johnson and McCulloh~\cite{johnson:2009vn} demonstrated a mathematical technique to convert meta networks into single-mode networks without losing critical information.  Using this methodology, we were able to derive distant relationships between nodes as a series of basic matrix algebra operations on all five networks. The
result is an agent based social network of potential terrorist. Characteristics of the transformed networks of agent node class only can
be found in Table~\ref{datasetTable}.

\begin{table}
\caption{Network Datasets}
\label{datasetTable}
\begin{center}
\begin{tabular}{|l|l|l|l|l|}
\hline
Name  & Nodes & Edges & Density & Avg. Degree\\
\hline
\hline
Tanzania & $17$ & $29$  & $0.213$ & $3.412$\\
\hline
GenTerrorNw1 & $57$ & $162$ & $0.102$ & $5.684$  \\
\hline
GenTerrorNw2 & $102$ & $388$ & $0.0753$ & $7.608$ \\
\hline
GenTerrorNw3 & $105$ & $590$ & $0.108$ & $11.238$ \\
\hline
GenTerrorNw4 & $135$ & $556$  & $0.0615$ & $8.237$\\
\hline
\hline
URV E-Mail & $1,133$ & $5,541$ & $0.00864$ & $9.781$ \\
\hline
CA-NetSci & $1,463$ & $2,743$ & $0.00256$ & $3.750$ \\
\hline
\hline
\end{tabular}
\end{center}
\end{table}

\subsection{Increasing the Fragility of Networks}

In our experiments, we showed that our algorithm was able to significantly increase the network-wide degree centrality by removing nodes - hence increasing the $fragile$ function with respect to a given network.  In each of the five real-world terrorist networks that we examined, removal of only $12\%$ of nodes can increase the network-wide centrality between $17\%$ and $45\%$ (see Figures~\ref{fragile1}-\ref{fragile5}).  In Figure~\ref{fragile0} we show a visualization of how the \textbf{Tanzania} network becomes more ``star-like'' with subsequent removal of nodes by the greedy algorithm.

\begin{figure}[htbb]
    \begin{center}
        \includegraphics[width=.7\linewidth]{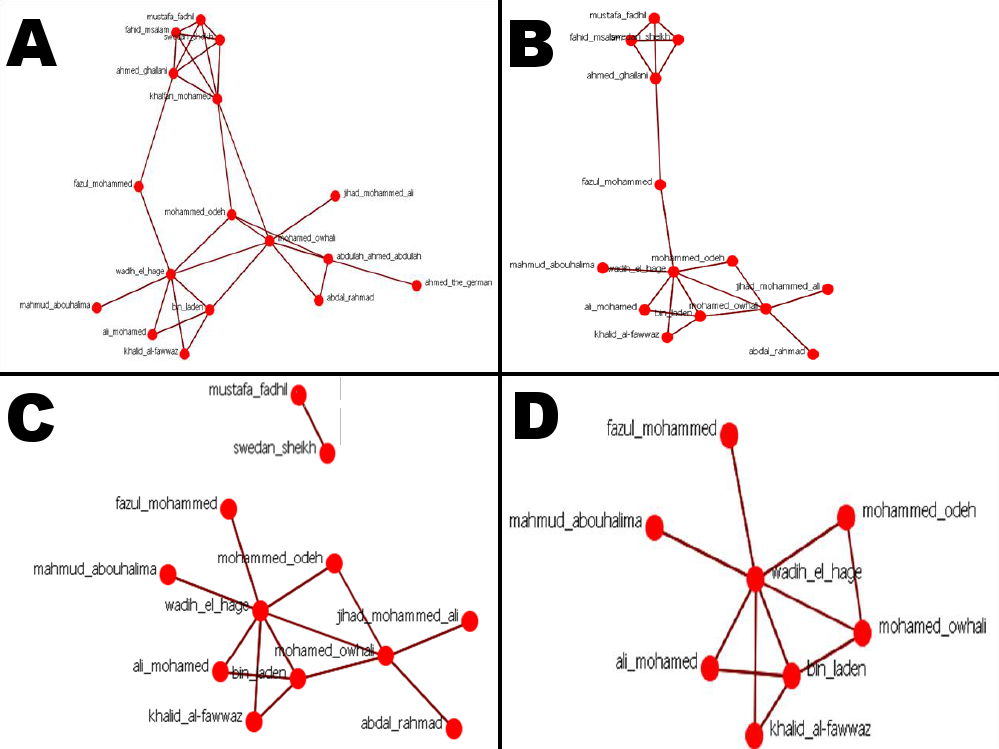}
    \end{center}
    \caption{Visualization of the \textbf{Tanzania} network after nodes removed by $\textsf{GREEDY}\_\textsf{FRAGILE}$.  Panel \textbf{A} shows the original network.  Panel \textbf{B} shows the network after $3$ nodes are removed, panel \textbf{C} shows the network after $5$ nodes are removed, and panel \textbf{D} shows the network after $9$ nodes are removed.  Notice that the network becomes more ``star-like'' after subsequent node removals.  In our experiment, after $\textsf{GREEDY}\_\textsf{FRAGILE}$ removed $11$ of the nodes in the network, it took the topology of a star.}
    \label{fragile0}
\end{figure}

For comparison, we also looked at the removal of high degree, closeness, and betweenness nodes.  Removal of high-degree, closeness, or betweenness nodes tended to increase the network-wide centrality.  In other words, traditional efforts of targeting leadership \textit{without} first conducting shaping operations may actually \textit{increase} the organization's ability to regenerate leadership - as such targeting operations effectively cause an organization to de-centralize.  We display these results graphically in Figures~\ref{fragile1}-\ref{fragile5}.  Notice that $\textsf{GREEDY}\_\textsf{FRAGILE}$ consistently causes an increase in the network-wide degree centrality.  An analysis of variance (ANOVA) reveals that there is a significant difference in the performance among our algorithm and the centrality measures with respect to increase or decrease in network-wide degree centrality ($p$-value less than $2.2 \cdot 10^{-16}$, calculated with \textbf{R} version 2.13).  Additionally, pairwise analysis conducted using Tukey's Honest Significant Difference (HSD) test indicates that the results of our algorithm differ significantly from any of the three centrality measures with a probability approaching $1.0$ ($95\%$ confidence, calculated with \textbf{R} version 2.13).  Typically, the ratio of percent increase in fragility to the percent of removed nodes is typically $2:1$ or greater.

\begin{figure}[htbb]
    \begin{center}
        \includegraphics[width=.8\linewidth]{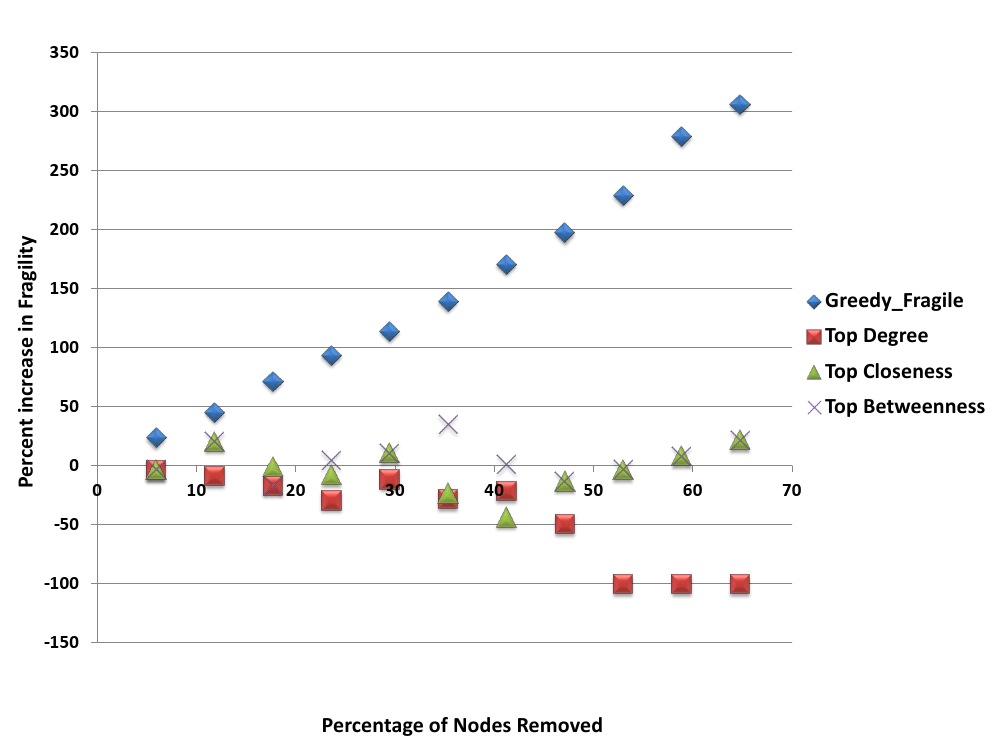}
    \end{center}
    \caption{Percent of nodes removed vs. percent increase in fragility for the \textbf{Tanzania} network using $\textsf{GREEDY}\_\textsf{FRAGILE}$, top degree, top closeness, and top betweenness.  The scale of the x-axis is positioned at $0\%$.}
    \label{fragile1}
\end{figure}

\begin{figure}[htbb]
    \begin{center}
        \includegraphics[width=0.8\linewidth]{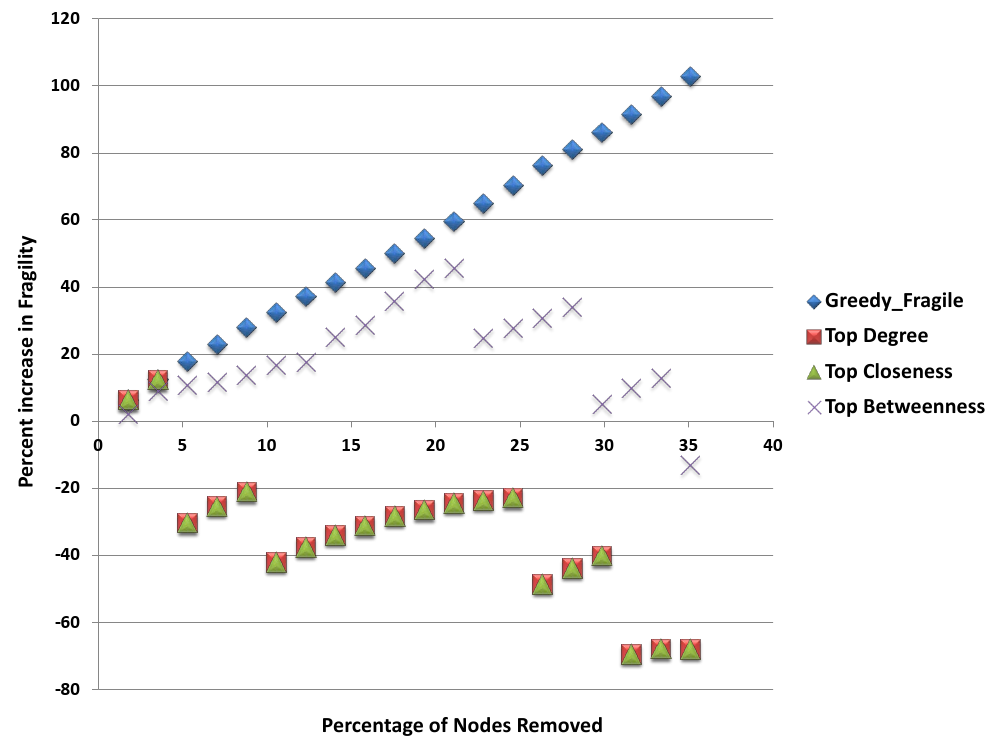}
    \end{center}
    \caption{Percent of nodes removed vs. percent increase in fragility for the \textbf{GenTerrorNet1} network using $\textsf{GREEDY}\_\textsf{FRAGILE}$, top degree, top closeness, and top betweenness.  The scale of the x-axis is positioned at $0\%$.}
    \label{fragile2}
\end{figure}

\begin{figure}[htbb]
    \begin{center}
        \includegraphics[width=.8\linewidth]{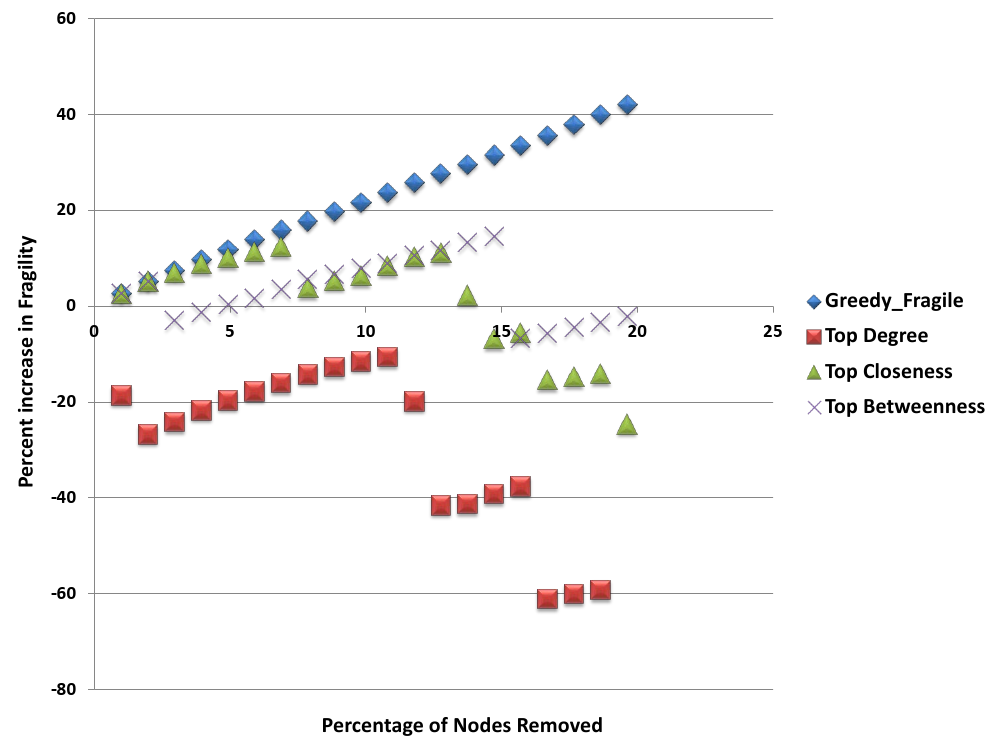}
    \end{center}
    \caption{Percent of nodes removed vs. percent increase in fragility for the \textbf{GenTerrorNw2} network using $\textsf{GREEDY}\_\textsf{FRAGILE}$, top degree, top closeness, and top betweenness.  The scale of the x-axis is positioned at $0\%$.}
    \label{fragile3}
\end{figure}

\begin{figure}[htbb]
    \begin{center}
        \includegraphics[width=.8\linewidth]{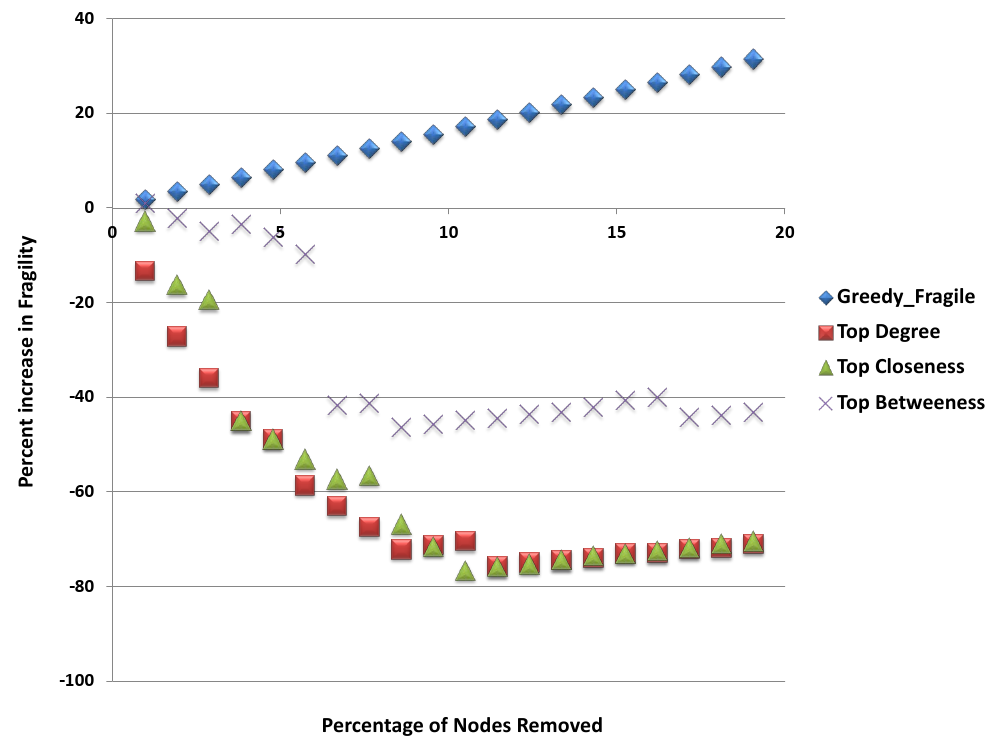}
    \end{center}
    \caption{Percent of nodes removed vs. percent increase in fragility for the \textbf{GenTerrorNw3} network using $\textsf{GREEDY}\_\textsf{FRAGILE}$, top degree, top closeness, and top betweenness.  The scale of the x-axis is positioned at $0\%$.}
    \label{fragile4}
\end{figure}

\begin{figure}[htbb]
    \begin{center}
        \includegraphics[width=.8\linewidth]{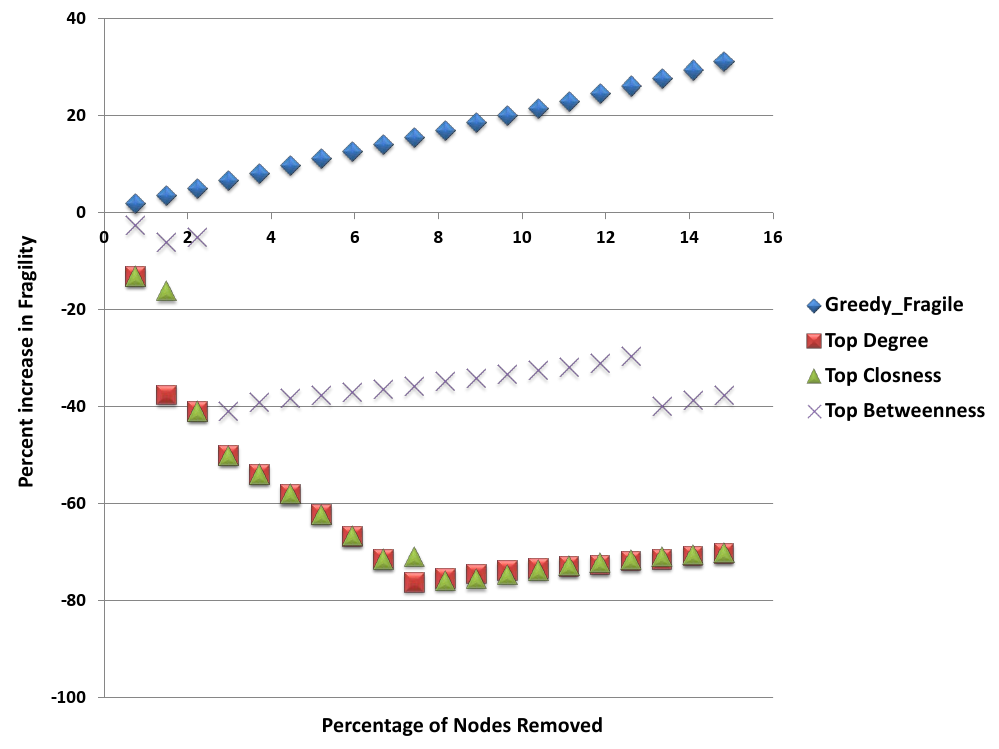}
    \end{center}
    \caption{Percent of nodes removed vs. percent increase in fragility for the \textbf{GenTerrorNw4} network using $\textsf{GREEDY}\_\textsf{FRAGILE}$, top degree, top closeness, and top betweenness.  The scale of the x-axis is positioned at $0\%$.}
    \label{fragile5}
\end{figure}

\subsection{Runtime}

We also evaluated the run-time of the $\textsf{GREEDY}\_\textsf{FRAGILE}$ algorithm.  With the largest terror network considered (GenTerrorNw4), we achieved short runtime (under $7$ seconds) on standard commodity hardware (see Figure~\ref{runtime}).  Hence, in terms of runtime, our algorithm is practical for use by a real-world analyst.  As predicted in our time complexity result, we found that the runtime of $\textsf{GREEDY}\_\textsf{FRAGILE}$ increases with the number of nodes removed.  We note that the implementations of top degree, closeness, and betweenness calculate those measures for the entire network at once - hence increasing the number of nodes to remove does not affect their runtime.

\begin{figure}[htbb]
    \begin{center}
        \includegraphics[width=.8\linewidth]{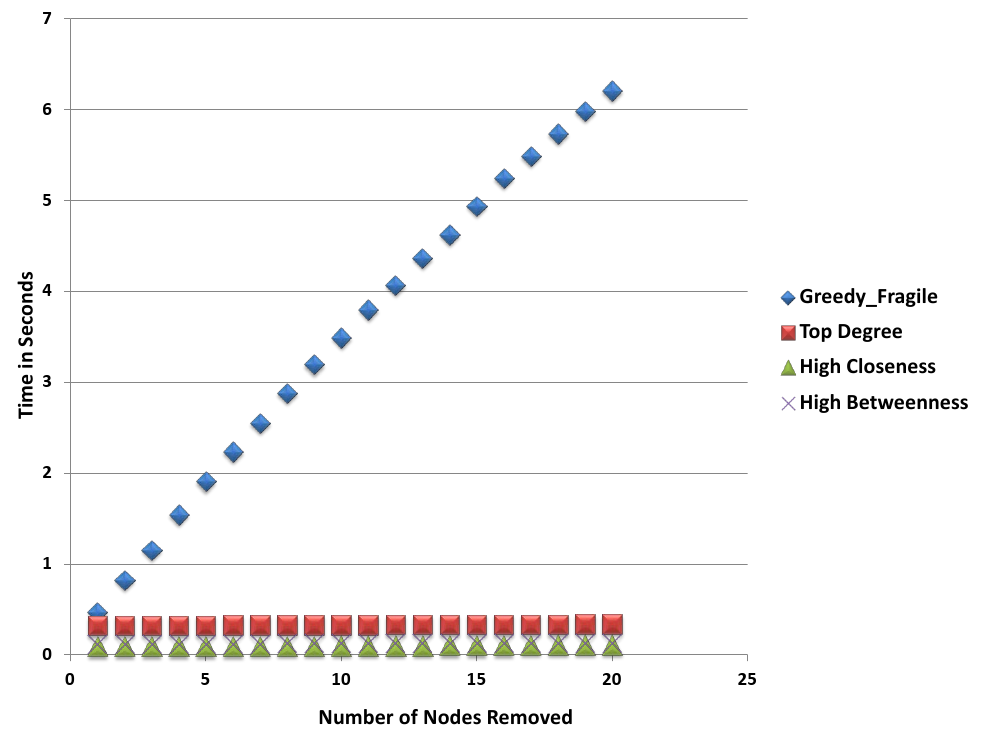}
    \end{center}
    \caption{Number of nodes removed vs. runtime for the \textbf{GenTerrorNw4} network using $\textsf{GREEDY}\_\textsf{FRAGILE}$, top degree, top closeness, and top betweenness.}
    \label{runtime}
\end{figure}

\subsection{Experiments on Large Data-Sets}
To study the scalability of $\textsf{GREEDY}\_\textsf{FRAGILE}$, we also employed it on two large social networks.  Note that these datasets are not terrorist or insurgent networks.  However, the larger size of these datasets is meant to illustrate how well our approach scales.  For these experiments, we used an e-mail network from University Rovira i Virgili (URV E-Mail)~\cite{uvi} and a Network Science collaboration network (CA-NetSci) from \cite{umich} (see Table~\ref{datasetTable}).  In Figure~\ref{scale1} we show the percentage of nodes removed vs. the percent increase in fragility.  We note that $2:1$ ratio of percent increase in fragility to the percent of removed nodes appears to be maintained even in these large datasets.  In Figure~\ref{scale2} we show the runtime for $\textsf{GREEDY}\_\textsf{FRAGILE}$ on the two large networks.  We note that the behavior of runtime vs. number of nodes removed resembles that of the GenTerrorNw4 network from the previous section.  Also of interest is that the algorithm was able to handle networks of over a thousand nodes in about $20$ minutes on commodity hardware.

\begin{figure}[htbb]
    \begin{center}
        \includegraphics[width=.8\linewidth]{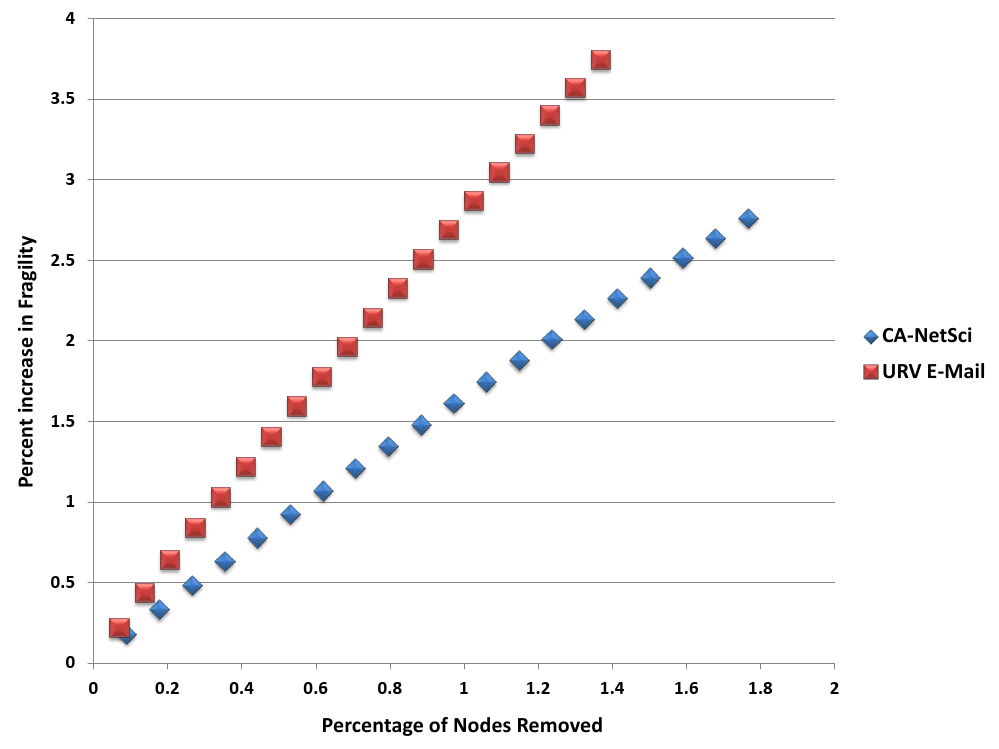}
    \end{center}
    \caption{Percent of nodes removed vs. percent increase in fragility for the \textbf{URV E-Mail} and \textbf{CA-NetSci} networks using $\textsf{GREEDY}\_\textsf{FRAGILE}$.}
    \label{scale1}
\end{figure}

\begin{figure}[htbb]
    \begin{center}
        \includegraphics[width=.8\linewidth]{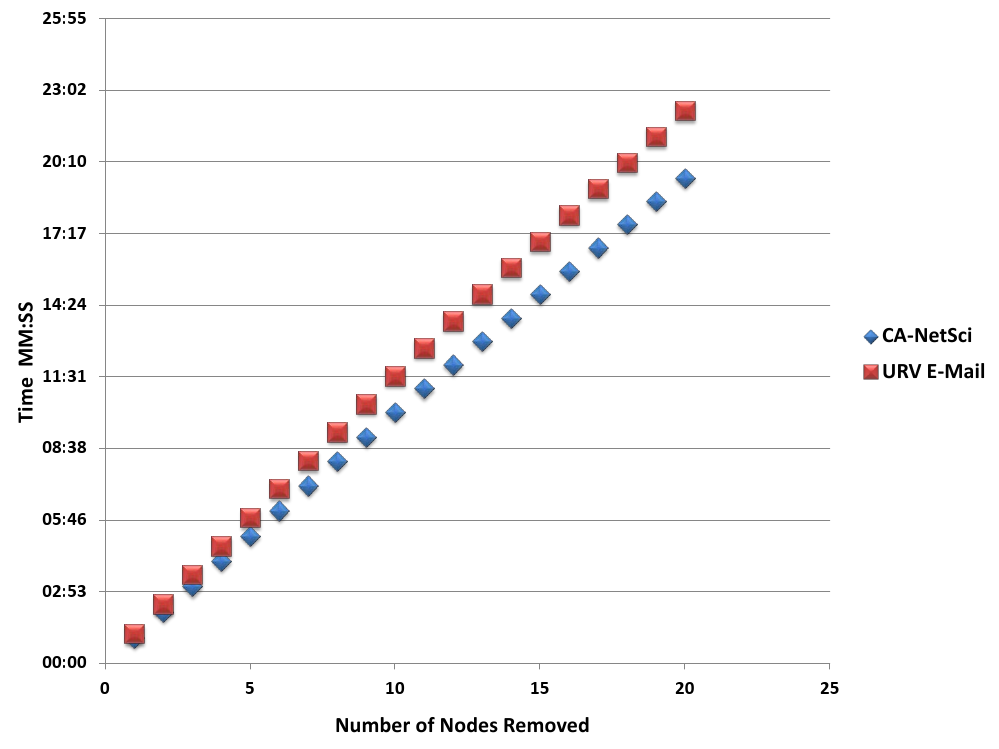}
    \end{center}
    \caption{Number of nodes removed vs. runtime for the \textbf{URV E-Mail} and \textbf{CA-NetSci} networks using $\textsf{GREEDY}\_\textsf{FRAGILE}$.}
    \label{scale2}
\end{figure}

\section{Related Work}
\label{rwSec}
Various aspects of the resiliency of terrorist networks have been previously explored in the literature.  For instance, \cite{Lindelauf2009126} studies the ability such network to facilitate communication while maintaining secrecy while \cite{gut10} studies how such networks are resilient to cascades.  However, to our knowledge, the network-wide degree centrality in such networks - and how to increase this property - has not been previously studied.

There has been much work dealing with the removal of nodes from a network to maximize fragmentation~\cite{albert,borgatti,arulselvan} where the nodes removed are mean to either increase fragmentation of the network or reduce the size of the largest connected component.  While this work has many applications, it is important to note that there are special considerations of terrorist and insurgent networks that we must account for in a targeting strategy.  For instance, if conducting a counter-intelligence operation while targeting, as in the case of \cite{slowBurn}, it may be desirable to preserve some amount of connectivity in the network.  Additionally, fragmentation of a network may result in the splintering of an organization into smaller, but more radical and deadly organizations.  This happens because in some cases, it may be desirable to keep certain terrorist or insurgent leaders in place to restrain certain, more radical elements of their organization.  Such splinter was observed for the insurgent organization Jaysh al-Mahdi in Iraq~\cite{uw-sg3}.  Further, these techniques do not specifically address the issue of emerging leaders.  Hence, if they were to be used for counter-terrorism or counter-insurgency, they would likely still benefit from a shaping operation to reduce organization's ability to regenerate leadership.

There has been some previous work on identifying emerging leaders in terrorist networks.  Although such an approach could be useful in identifying certain leaders, it does not account the organizations ability as a whole to regenerate leadership.  In \cite{carley04}, the topic of \textit{cognitive demand} is studied.  The cognitive load of an individual deals with their ability to handle multiple demands on their time and work on complex tasks.  Typically, this can be obtained by studying networks where the nodes may represent more than individual people - but tasks, events, and responsibilities.  However, it may often be the case that this type of information is often limited or non-existent in many situations.  Additionally, as discussed throughout this paper, the targeting of individual nodes may often not be possible for various reasons.  Hence, our framework, that focuses on the \textit{network's ability to regenerate leadership} as opposed to finding individual emerging leaders may be more useful as we can restrict the available nodes in our search using the ``no strike list.''    By removing these nodes from targeting consideration - but by still considering their structural role - our framework allows a security force to reduce the regenerative ability of a terror network by ``working around'' individuals that may not be targeted.

In more recent work \cite{petersen11} looks at the problem of removing leadership nodes from a terrorist or criminal network in a manner that accounts for new links created in the aftermath of an operation.  Additionally, \cite{ovelgonne12} look at identifying leaders in covert terrorist network who attempt to minimize their communication due to the clandestine nature of their operations.  They do this by introducing a new centrality measure called ``covertness centrality.''  Both of these approaches are complementary to ours as they focus on the leadership of the terrorist or insurgent group - as this approach focuses on the networks ability to re-generate leadership.  A more complete integration of this approach leadership targeting method such as these (i.e. using a network-wide version of covertness centrality) is an obvious direction for future work.

\section{Conclusions}

In this paper we described how to target nodes in a terrorist or insurgent network as part of a \textit{shaping} operation designed to reduce the organization's ability to regenerate leadership.  Our key intuition was to increase the network-wide degree centrality which would likely have the effect of eliminating emerging leaders as maximizing this quantity would intuitively increase the organization's reliance on a single leader.  In this paper, we found that though identifying a set of nodes to maximize this network-wide degree centrality is NP-hard, our greedy approach proved to be a viable heuristic for this problem, increasing this quantity between $17\%-45\%$ in our experiments.  Future work could include an examination of other types of network-wide centrality -- for instance network-wide closeness centrality -- instead of network-wide degree centrality.  Another aspect that we are considering in ongoing research is determining the effectiveness of the shaping strategy when we have observed only part of the terrorist or insurgent organization -- as is often the case as such networks are created from intelligence data.

\section*{Acknowledgements}
We would like to thank Jon Bentley and Charles Weko for their feedback on an earlier version of this paper.\\
\indent Some of the authors are supported under by the Army Research Office (project 2GDATXR042).  The opinions in this paper are those of the authors and do not necessarily reflect the opinions of the funders, the U.S. Military Academy, or the U.S. Army.

\end{document}